\pdfoutput=1
\documentclass[a4paper,11pt]{article}
\usepackage[
top=30truemm,
bottom=30truemm,
left=25truemm,
right=25truemm
]{geometry} %Kira edit
\usepackage[dvips]{graphicx}
\usepackage{amsmath,amsthm,amssymb}
\usepackage{amsfonts}
\usepackage{color}
\usepackage{pifont}
\usepackage{here}
\usepackage{cases}
\usepackage{empheq}
\usepackage{multirow} %edit by Kira
\usepackage{colortbl} %edit by Kira
\usepackage{arydshln} %表で破線を使う
\usepackage{url} %edit by Kira
\usepackage[ %Kira edit
algo2e,
boxed,
vlined %delete ''end''
]{algorithm2e} %edit
\SetAlgoSkip{medskip} %edit by Kira for algorithm2e
\SetArgSty{textrm} %edit by Kira for algorithm2e
\IncMargin{1ex} %for algorithm2e
\SetProcNameSty{textsc} %for algorithm2e
\SetKwRepeat{Do}{do}{while}
\def\dfrac{\displaystyle\frac}

\newtheorem{thm}{Theorem}[section]%
\newtheorem{lem}{Lemma}[section]%
\begin{document}
%%%%%%%%%%%%%%%%%%%%%%%%       title       %%%%%%%%%%%%%%%%%%%%%%%%%%%%%%%%%%%%
\title{Fast enumeration of effective mixed transports \\ for recommending shipper collaboration}
\author{Akifumi~Kira $^{\tt a, *}$, Nobuo Terajima $^{\tt b}$ \\
}
\date{\empty}
\maketitle
\vspace{-10mm}
\begin{center}
\small
\begin{tabular}{l}
$^{\tt a}$ Institute of Mathematics for Industry, Kyushu University \\
\quad 744 Motooka, Nishi-ku, Fukuoka 819-0395, Japan \\
$^{\tt b}$ Japan Pallet Rental Corporation \\
\quad Ote Center Building, 1-1-3 Otemachi, Chiyoda-ku, Tokyo 100-0004, Japan \\
$^{\tt *}$ Corresponding author, E-mail: kira@imi.kyushu-u.ac.jp
\end{tabular}
\end{center}
%%%%%%%%%%%%%%%%%%%%%%%%      abstract     %%%%%%%%%%%%%%%%%%%%%%%%%%%%%%%%%%%%
\begin{abstract}
In this study, we focus on a form of joint transportation called mixed transportation 
and enumerate the combinations with high cooperation effects 
from among a number of transport lanes registered in a database (logistics big data). 
As a measure of the efficiency of mixed transportation, 
we consider the reduction rate that represents how much the total distance of loading trips is shortened by cooperation. The proposed algorithm instantly presents the set of all mixed transports 
with a reduction rate of a specified value or less.  
This algorithm 
is more than 7,000 times faster than simple brute force.
\end{abstract}
%%%%%%%%%%%%%%%%%%%%%%%%%%%%%%%%%%%%%%%%%%%%%%%%%%%%%%%%%%%%%%%%%%%%%%%%%%%%%%%

%%%%%%%%%%%%%%%%%%%%%%%%%%%%%%%%%%%%%%%%%%%%%%%%%%%%%%%%%%%%%%%%%%%%%%%%%%%%%%%
\section{Introduction}
%%%%%%%%%%%%%%%%%%%%%%%%%%%%%%%%%%%%%%%%%%%%%%%%%%%%%%%%%%%%%%%%%%%%%%%%%%%%%%%
While there has been concern about a logistics crisis and the shortage of drivers has been becoming increasingly serious, improving labor productivity in logistics has become an urgent issue. However, truck load efficiency remains low at less than 40\%~\cite{policy-b}.
In the General Principles of Logistics Policies~\cite{policy-a}, which provides the government's guidelines on logistics and logistics administration, improving the efficiency of logistics through collaboration is set as a goal, and there is strong demand for us to shift from specific optimization to total optimization or from competition to co-creation. 
\par
Kira et al.~\cite{Kira2023} study a form of joint transportation called triangular transportation, 
wherein three shipments are processed sequentially in the form of a triangle,  
and propose successful algorithms that enumerate the combinations with high cooperation effects.
The joint transportation matching system \lq\lq TranOpt'' provided by the Japan Pallet Rental Corporation (JPR) is a service that uses this patented technology to match cargo owner companies looking for joint transportation partners with one another and propose efficient joint transportation~\cite{JPR2021, Kira2021}
\footnote{Patent rights were jointly acquired by Gunma University and JPR~\cite{Kira2021}.}.  
This service is already used by more than 260 companies (as of November 2024)
\footnote{The authors won the Case Study Award 2023 of the Operations Research Society of Japan for this research.}. 
\par
Although Kira et al.~\cite{Kira2023} only deal with triangular transportation, TranOpt has 
another important function that supports the high-speed enumeration of mixed
transportation, in which loads are mixed and transported simultaneously (see Figure~\ref{fig:mixed_transport}). 
As a measure of the efficiency of mixed transportation, 
the reduction rate, which represents how much the total distance of loading trips is shortened by cooperation, 
is taken into account. 
In this study, we explain the logic installed in TranOpt in detail and prove its correctness.

%Therefore, the Japan Pallet Rental Corporation (JPR) has provided significant support  toward joint transport 
%by companies in other industries, as well as other measures to carry more cargo with fewer trucks. In October 2019, JPR recognized the importance of creating a system for deploying such initiatives throughout the logistics industry and 
%began developing a common transportation matching system using artificial intelligence (AI) technology in collaboration with Gunma University. From October 2019 to March 2021, this development was funded by the New Energy and Industrial Technology Development Organization (NEDO)\footnote{
%National Research and Development Agency under the jurisdiction of Japan's Ministry of Economy, Trade and Industry.
%}
% \lq\lq The Project to Promote Data-Sharing in Collaborative Areas and Developing of AI Systems to Promote Connected Industries).''

\begin{figure}[H]
\centering
\begin{tabular}{l}
$\longrightarrow$ : Loading Trip \\
$\dashrightarrow$\, : Empty Trip
\end{tabular}
\hspace{5mm}
$\textrm{Reduction Ratio} = \dfrac{\textrm{Length of} \longrightarrow \textrm{when cooperating}}{\textrm{Length of} \longrightarrow \textrm{when not cooperating}}$ \\[3mm]
\includegraphics[scale=0.14]{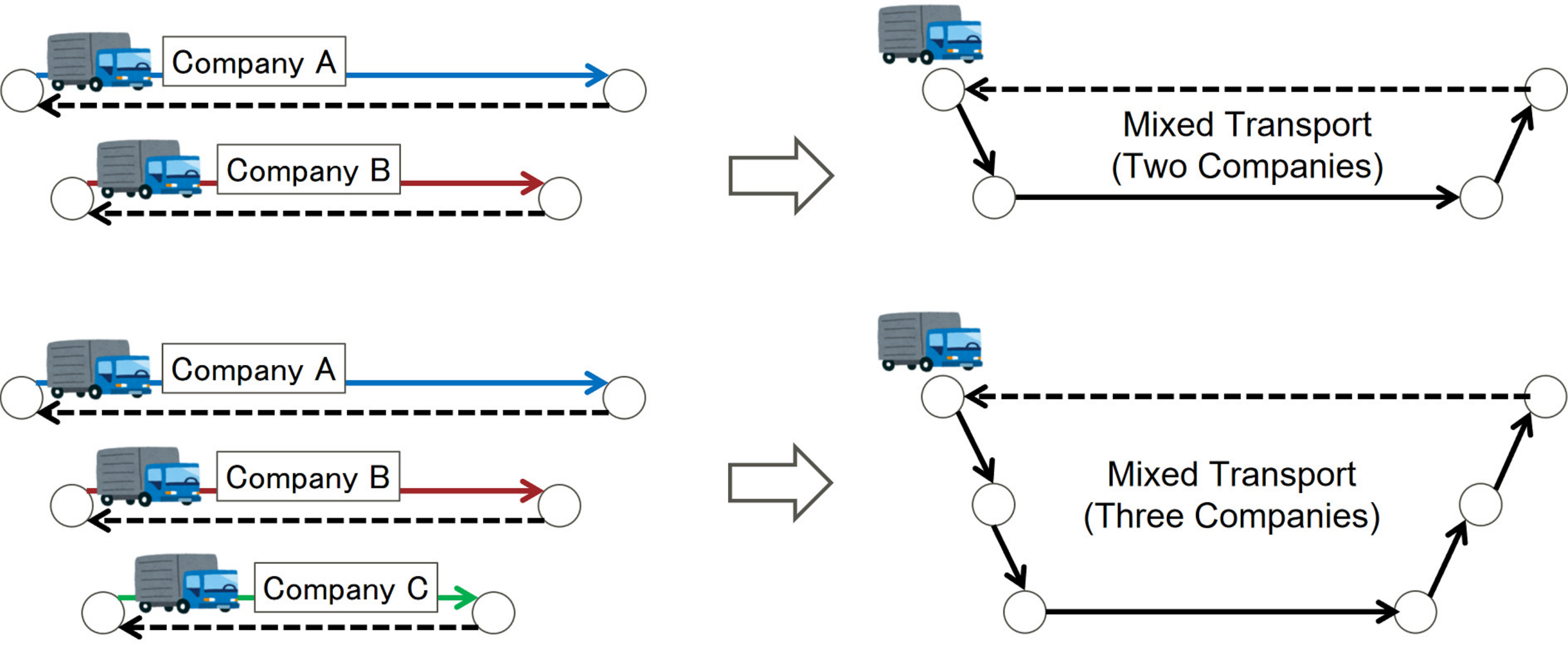}
\caption{Mixed transportation}
\label{fig:mixed_transport}
\end{figure}
%%%%%%%%%%%%%%%%%%%%%%%%%%%%%%%%%%%%%%%%%%%%%%%%%%%%%%%%%%%%%%%%%%%%%%%%%%%%%%%
%
\section{Problem Formulation} 
\label{Sec:formulation}
%
%%%%%%%%%%%%%%%%%%%%%%%%%%%%%%%%%%%%%%%%%%%%%%%%%%%%%%%%%%%%%%%%%%%%%%%%%%%%%%%

We are given a metric space $(B, d)$,  
where $B$ is a finite set of transportation bases and 
$d : B \!\times\! B \to \mathbb{R}_{\geq 0}$ is a distance function. 
For simplicity, a truckload request from one base to another is called a transport lane (or lane),   
and a finite set of lanes $T$ is provided. 
For every lane $t \in T$, we denote the start point and the endpoint as $t^{\rm s}$ and $t^{\rm e}$, respectively.
Furthermore, for every lane $t \in T$, we denote the distance of $t$ by $t^{\rm d}$, where $t^{\rm d} = d(t^{\rm s}, t^{\rm e})$.
\par
A series of transportation in which a single truck mixes the loads of multiple lanes 
$t_1, t_2,\, \ldots, t_n$ and transports them simultaneously is called a \lq\lq mixed transport'' and 
denoted by $(t_1, t_2,\,\ldots, t_n)$. 
This notation implies that loads are loaded in the order $t_1, t_2,\, \ldots t_n$ and 
unloaded in the reverse order $t_n,\, \ldots, t_2, t_1$. 
This is due to the last-in, first-out (LIFO) nature of truck beds. 
This study focuses on the case $n = 3$. 
Given a mixed transport $(t_1, t_2, t_3)$, we use symbols $d_i, x_i, x, z_i$, and $z$ in the following manner 
(see also Figure~\ref{fig:d_x_and_z}):
\begin{align*}
d_i &= t_i^{\rm d}, \quad i = 1,2,3,  \\
x_1 &= d(t_{1}^{\rm s}, t_{2}^{\rm s}), \quad x_2 = d(t_{2}^{\rm s}, t_{3}^{\rm s}), \\
z_1 &= d(t_{2}^{\rm e}, t_{1}^{\rm e}), \quad z_2 = d(t_{3}^{\rm e}, t_{2}^{\rm e}), \\
x &= d(t_{1}^{\rm s}, t_{3}^{\rm s}), \quad y = d(t_{3}^{\rm s}, t_{1}^{\rm e}), \quad z = d(t_{3}^{\rm e}, t_{1}^{\rm e}), \\
\end{align*}

\begin{figure}[H]
\centering
\includegraphics[scale=0.11]{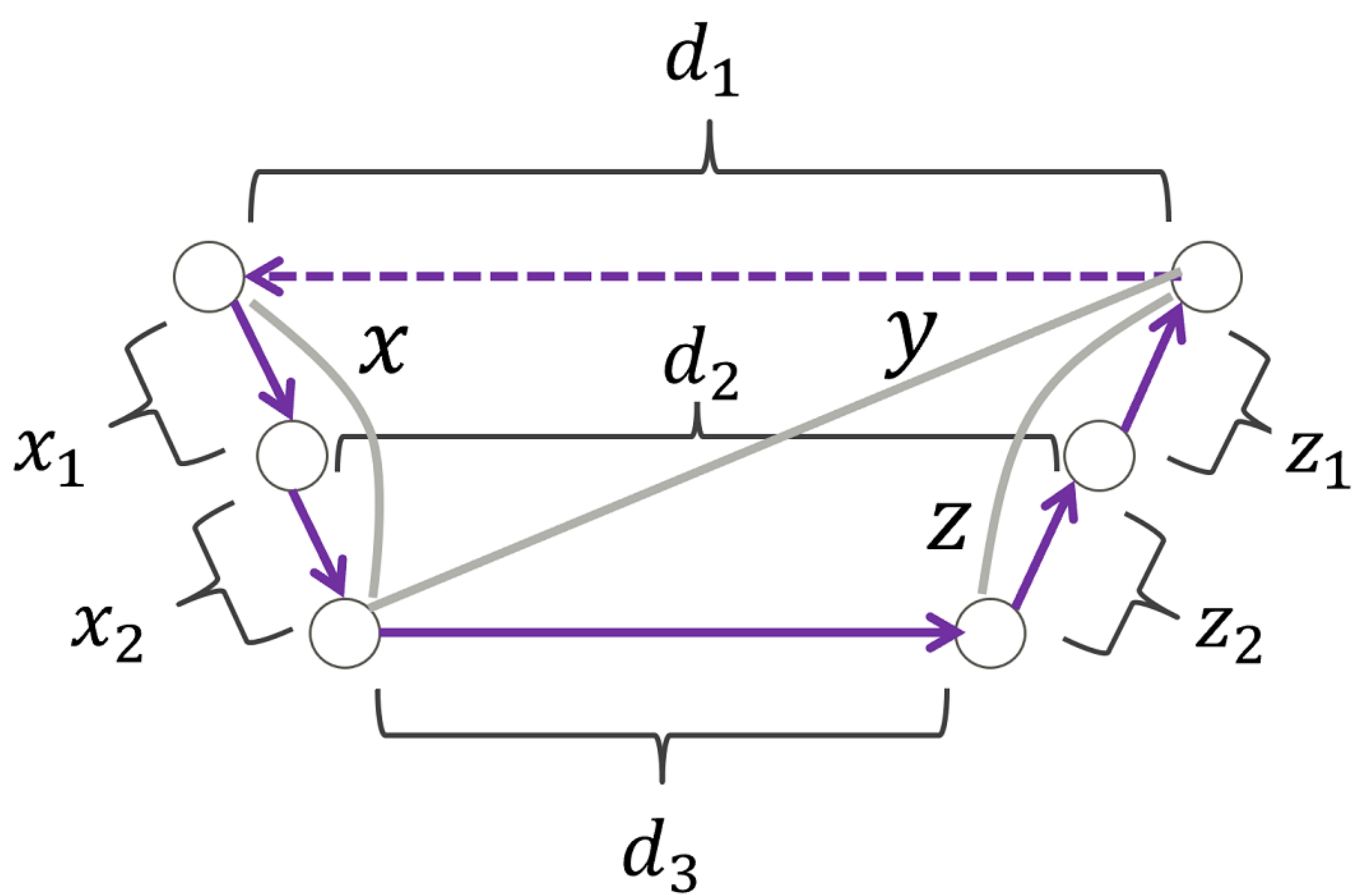} 
\caption{Symbols $d_i, x_i, z_i, x, y, z$ representing distances}
\label{fig:d_x_and_z}
\end{figure}

As a measure of the efficiency of mixed transportation, 
it is natural to consider the following reduction rate:
\begin{equation*}
\textrm{Reduction rate} = \frac{x_1 + x_2 + d_3 + z_2 + z_1}{d_1 + d_2 + d_3}.
\end{equation*}
The reduction rate represents how much the total distance of loading trips is shortened by cooperation; 
the smaller this value is, the more efficient the mixed transport is. 
We note that 
the reduction rate is at its minimum when $t_1^{\rm s} = t_2^{\rm s} = t_3^{\rm s}$ and $t_1^{\rm e} = t_2^{\rm e} = t_3^{\rm e}$, 
and in that case the minimum value is $\frac{1}{3}$.

\par
For a given lane $t_1 \in T$, we search for joint transport partners $t_2, t_3 \in T$ 
and propose an effective mixed transport $(t_1, t_2, t_3)$ (see Figure 3). 
In particular, for any specified $r \in [\frac{1}{3}, 1)$, our goal is to list all the mixed transports with a reduction rate of $r$ or less. 

\begin{figure}[H]
\centering
\includegraphics[scale=0.15]{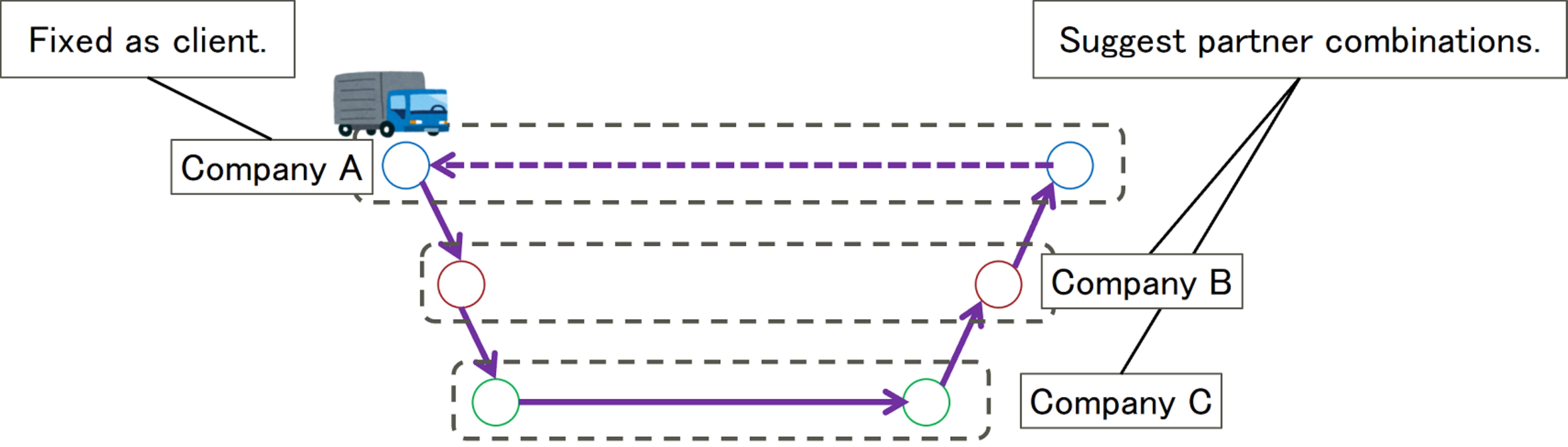} 
\caption{Processing a mixed transport matching request}
\label{fig:system}
\end{figure}

For a given transport lane $t_1$, we only need to search for partners $t_2$ and $t_3$. 
Therefore we can consider a simple brute-force search with a double for-loop, 
as shown in Algorithm~\ref{Alg:simple}. 

\begin{algorithm2e}[h]
\SetKwData{Input}{input}
\KwData{a set of lanes $T$ on a metric space $(B, d)$, a lane $t_1$, and a desired reduction rate $r \in [\frac{1}{3}, 1)$}
\KwResult{the set of all mixed transports starting from $t_1^{\rm s}$ with a reduction rate of $r$ or less}

$C \leftarrow \emptyset$\;
$d_1 \leftarrow t_1^{\rm d}$\;
\ForAll{$t_2 \in T \setminus \{t_1\}$}{
    $d_2 \leftarrow t_2^{\rm d}$\;
    $x_1 \leftarrow d(t_1^{\rm s}, t_2^{\rm s})$\;
    $z_1 \leftarrow d(t_2^{\rm e}, t_1^{\rm e})$\;
    \ForAll{$t_3 \in T \setminus \{t_1, t_2\}$}{
        $d_3 \leftarrow t_3^{\rm d}$\;
        $x_2 \leftarrow d(t_2^{\rm s}, t_3^{\rm s})$\;
        $z_2 \leftarrow d(t_3^{\rm e}, t_2^{\rm e})$\;
        \If{$\frac{x_1 + x_2 + d_3 + z_2 + z_1}{d_1 + d_2 + d_3} \leqq r$}{
            $C \leftarrow C \cup \{(t_1, t_2, t_3)\}$\;
        }
    }
}
\caption{A simple brute-force search}
\label{Alg:simple}
\end{algorithm2e}
%%%%%%%%%%%%%%%%%%%%%%%%%%%%%%%%%%%%%%%%%%%%%%%%%%%%%%%%%%%%%%%%%%%%%%%%%%%%%%%
%
\section{Pruning Algorithm} 
\label{Sec:pruning}
%
%%%%%%%%%%%%%%%%%%%%%%%%%%%%%%%%%%%%%%%%%%%%%%%%%%%%%%%%%%%%%%%%%%%%%%%%%%%%%%%
For any transportation base $b \in B$, let $T(b)$ denote the set of all lanes 
starting from $b$. In addition, let $S$ be the set of all bases that can be the starting point for a lane. 
That is,
\begin{equation*}
S := \{ b \in B \,|\, T(b) \neq \emptyset \}.
\end{equation*}

Our proposed algorithm is presented as in Algorithm~\ref{Alg:pruning}.

\begin{algorithm2e}[H]
\SetKwData{Input}{input}
\KwData{a set of lanes $T$ on a metric space $(B, d)$, a lane $t_1$, and a desired reduction rate $r \in [\frac{1}{3}, 1)$}
\KwResult{the set of all mixed transports starting from $t_1^{\rm s}$ with a reduction rate of $r$ or less}

$C \leftarrow \emptyset$\;
$d_1 \leftarrow t_1^{\rm d}$\;
\ForAll{$s \in S$ such that $d(t_1^{\rm s}, s) + d(s, t_1^{\rm e}) \leqq \frac{2r}{1-r}d_1$}{
    $x \leftarrow d(t_1^{\rm s}, s)$\;
    \ForAll{$t_3 \in T(s) \setminus \{t_1\}$ such that $(1 - 2r) t_3^{\rm d} + (1 - r) d(t_3^{\rm e}, t_1^{\rm e}) \leqq (r - 1) x + r d_1$}{
        $d_3 \leftarrow t_3^{\rm d}$\;
        $z \leftarrow d(t_3^{\rm e}, t_1^{\rm e})$\;
        \ForAll{$s' \in S$ such that $d(t_1^{\rm s}, s') + (1 - r) d(s', t_3^{\rm s}) \leqq rd_1 + (2r - 1) d_3 - (1 - r) z$}{
            $x_1 \leftarrow d(t_1^{\rm s}, s')$\;
            $x_2 \leftarrow d(s', t_3^{\rm s})$\;
            \ForAll{$t_2 \in T(s') \setminus \{t_1, t_3\}$ such that $t_2^{\rm d} \geqq \frac{1}{r}(x_1 + x_2 + z) - d_1 + \frac{1 - r}{r} d_3$}{
                $d_2 \leftarrow t_2^{\rm d}$\;
                $z_2 \leftarrow d(t_3^{\rm e}, t_2^{\rm e})$\;
                $z_1 \leftarrow d(t_2^{\rm e}, t_1^{\rm e})$\;
                \If{$\frac{x_1 + x_2 + d_3 + z_2 + z_1}{d_1 + d_2 + d_3} \leqq r$}{
                    $C \leftarrow C \cup \{(t_1, t_2, t_3)\}$\;
                }
            }
        }
    }
}
\caption{A brute-force search with pruning}
\label{Alg:pruning}
\end{algorithm2e}
%If we partition $T$ into $T = \cup_{s \in S} T(s)$ in advance, then 
%the for loop traversing $T$ can be replaced by a double for loop 
%(outer for loop traversing $S$ and inner for loop traversing $T(s)$).
%Thus, Algorithm~\ref{Alg:simple}, which is written as a double for loop, 
%can be rewritten as a quadruple for loop, as shown in Algorithm~\ref{Alg:quadruple}.

%\begin{algorithm2e}[H]
%\SetKwData{Input}{input}
%\KwData{a set of lanes $T$ on a metric space $(B, d)$, a lane $t_1$, and a desired reduction rate $r \in [\frac{1}{3}, 1)$}
%\KwResult{the set of all mixed transports starting from $t_1^{\rm s}$ with a reduction rate of $r$ or less}
%
%$C \leftarrow \emptyset$\;
%$d_1 \leftarrow t_1^{\rm d}$\;
%\ForAll{$s \in S$}{
%    \ForAll{$t_3 \in T(s) \setminus \{t_1\}$}{
%        $d_3 \leftarrow t_3^{\rm d}$\;
%        \ForAll{$s' \in S$}{
%            $x_1 \leftarrow d(t_1^{\rm s}, s')$\;
%            $x_2 \leftarrow d(s', t_3^{\rm s})$\;
%            \ForAll{$t_2 \in T(s') \setminus \{t_1, t_3\}$}{
%                $d_2 \leftarrow t_2^{\rm d}$\;
%                $z_2 \leftarrow d(t_3^{\rm e}, t_2^{\rm e})$\;
%                $z_1 \leftarrow d(t_2^{\rm e}, t_1^{\rm e})$\;
%                \If{$\frac{x_1 + x_2 + d_3 + z_2 + z_1}{d_1 + d_2 + d_3} \leqq r$}{
%                    $C \leftarrow C \cup \{(t_1, t_2, t_3)\}$\;
%                }
%            }
%        }
%    }
%}
%\caption{Quadruple looping of brute-force search}
%\label{Alg:quadruple}
%\end{algorithm2e}

In order to prove the validity of Algorithm~\ref{Alg:pruning}, we introduce four lemmas.

\begin{lem}[pruning for the first for-loop]
A necessary condition for a mixed transport to 
have an reduction rate of $r$ or less is
\begin{equation*}
x +y \leqq \frac{2r}{1-r}d_1.
\end{equation*}
\label{Lem:pruning1}
\end{lem}

\begin{proof}
We can find a lower bound of the reduction rate as follows:
\begin{align*}
r &\geqq \frac{x_1 + x_2 + d_3 + z_2 + z_1}{d_1 + d_2 + d_3}   \quad \textrm{(monotonically decreasing w.r.t. $d_2$)} \\
&\geqq \frac{x_1 + x_2 + d_3 + z_2 + z_1}{d_1 + (x_1 + d_1 + z_1) + d_3}  
= \frac{x_1 + x_2 + d_3 + z_2 + z_1}{x_1 + 2d_1 + d_3 + z_1} \\
&\geqq \frac{x_1 + x_2 + d_3 + z_2 + z_1}{x_1 + x_2 + 2d_1 + d_3 + z_2 + z_1} \quad \textrm{(monotonically increasing w.r.t. $x_1 + x_2, d_3 + z_2 + z_1$)} \\
&\geqq \frac{x + y}{x + y + 2d_1}.
\end{align*}
In the above, to obtain the second inequality, we replace $d_2$ with $x_1 + d_1 + z_1$, 
which is an upper bound of $d_2$. Similarly, to obtain the last inequality, we replace
$x_1 + x_2$ and $d_3 + z_2 + z_1$ with their respective lower bounds $x$ and $y$.
These upper and lower bounds follow from the metric axioms. 
Thus, by solving $r \geqq \frac{x + y}{x + y +2d_1}$ with respect to $x + y$, we obtain the result.
\end{proof}

\begin{lem}[pruning for the second for-loop]
A necessary condition for a mixed transport to 
have an reduction rate of $r$ or less is
\begin{equation*}
(1 - 2r) d_3 + (1 - r) z \leqq (r - 1) x + r d_1.
\end{equation*}
\label{Lem:pruning2}
\end{lem}

\begin{proof}
We can find a lower bound of the reduction rate as follows:
\begin{align}
r &\geqq \frac{x_1 + x_2 + d_3 + z_2 + z_1}{d_1 + d_2 + d_3}  \quad \textrm{(monotonically decreasing w.r.t. $d_2$)} \nonumber \\
&\geqq \frac{x_1 + x_2 + d_3 + z_2 + z_1}{d_1 + (x_2 + d_3 + z_2) + d_3}  
= \frac{x_1 + x_2 + d_3 + z_2 + z_1}{x_2 + d_1 + 2d_3 + z_2}  \label{ineq-a} \\
&\geqq \frac{x_1 + x_2 + d_3 + z_2 + z_1}{x_1 + x_2 + d_1 + 2d_3 + z_2 + z_1} \quad \textrm{(monotonically increasing w.r.t. $x_1 + x_2, z_2 + z_1$)} \nonumber \\
&\geqq \frac{x + z + d_3}{x + z + d_1 + 2d_3} \nonumber
\end{align}
In the above, to obtain the second inequality, we replace $d_2$ with $(x_2 + d_3 + z_2)$, 
which is an upper bound of $d_2$. 
Similarly, to obtain the last inequality, we replace
$x_1 + x_2$ and $z_2 + z_1$ with their respective lower bounds $x$ and $z$.
From the final inequality $r \geqq \frac{x + z + d_3}{x + z + d_1 + 2d_3}$, we obtain the result\footnote{
By solving the final inequality with respect to $d_3$, we have
\begin{equation*}
\left\{
\begin{array}{ll}
d_3 \leqq - \dfrac{1-r}{1-2r} (x + z) + \frac{r}{1-2r}d_1 & \textrm{if $r \in [\frac{1}{3}, \frac{1}{2})$,} \\[3mm]
d_3 > 0 & \textrm{if $r = \frac{1}{2}$ and $x = d_1$,} \\[1mm]
\textrm{infeasible} & \textrm{if $r = \frac{1}{2}$ and $x \neq d_1$,} \\[1mm]
d_3 \geqq \dfrac{1-r}{2r-1} (x + z) - \frac{r}{2r-1}d_1 & \textrm{if $r \in [\frac{1}{2}, 1)$.} \\[1mm]
\end{array}
\right.
\end{equation*}
}.
\end{proof}

\begin{lem}[pruning for the third for-loop]
A necessary condition for a mixed transport to 
have an reduction rate of $r$ or less is
\begin{equation*}
x_1 + (1 - r) x_2 \leqq rd_1 + (2r - 1) d_3 - (1 - r) z.
\end{equation*}
\label{Lem:pruning3}
\end{lem}

\begin{proof}
From (\ref{ineq-a}), we can find a lower bound of the reduction rate as follows:
\begin{align*}
r &\geqq \frac{x_1 + x_2 + d_3 + z_2 + z_1}{x_2 + d_1 + 2d_3 + z_2} \\
&\geqq \frac{x_1 + x_2 + d_3 + z_2 + z_1}{x_2 + d_1 + 2d_3 + z_2 + z_1} \quad \textrm{(monotonically increasing w.r.t. $z_2 + z_1$)} \\
&\geqq \frac{x_1 + x_2 + d_3 + z}{x_2 + d_1 + 2d_3 + z}.
\end{align*}
The result follows from the final inequality $r \geqq \frac{x_1 + x_2 + d_3 + z}{x_2 + d_1 + 2d_3 + z}.$.
\end{proof}

\begin{lem}[pruning for the fourth for-loop]
A necessary condition for a mixed transport to 
have an reduction rate of $r$ or less is
\begin{equation*}
d_2 \geqq \frac{1}{r}(x_1 + x_2 + z) - d_1 + \frac{1 - r}{r} d_3.
\end{equation*}
\label{Lem:pruning4}
\end{lem}

\begin{proof}
We can find a lower bound of the reduction rate as follows:
\begin{align*}
r &\geqq \frac{x_1 + x_2 + d_3 + z_2 + z_1}{d_1 + d_2 + d_3}  \quad \textrm{(monotonically decreasing w.r.t. $z_2 + z_1$)} \\
&\geqq \frac{x_1 + x_2 + d_3 + z}{d_1 + d_2 + d_3}. 
\end{align*}
By solving the final inequality $r \geqq \frac{x_1 + x_2 + d_3 + z}{d_1 + d_2 + d_3}$ with respect to $d_2$, we obtain the result.
\end{proof}

\begin{thm}
Algorithm~\ref{Alg:pruning} correctly outputs the set of all mixed transports starting from $t_1^{\rm s}$ with a reduction rate of $r$ or less. 
\end{thm}

\begin{proof}
The result follows from Lemmas~\ref{Lem:pruning1} to \ref{Lem:pruning4}.
\end{proof}

By keeping track of the reduction rate, which is the provisional $k$th rank,  
we can improve Algorithm~\ref{Alg:pruning} to obtain the k-best solutions more instantly. 
To do this, we use a binary heap~\cite{heap}.
The improved algorithm is presented as Algorithm~\ref{Alg:dynamic}.

\begin{algorithm2e}[H]
\SetKwData{Input}{input}
\KwData{a set of lanes $T$ on a metric space $(B, d)$, a lane $t_1$, a desired reduction rate $r \in [\frac{1}{3}, 1)$, and an integer $k$}
\KwResult{the set of $k$ mixed transports starting from $t_1^{\rm s}$ with a reduction rate of $r$ or less, starting with those with the
smallest reduction rate.}
Declare a binary heap (max-heap) $H$\;
$d_1 \leftarrow t_1^{\rm d}$\;
\ForAll{$s \in S$ such that $d(t_1^{\rm s}, s) + d(s, t_1^{\rm e}) \leqq \frac{2r}{1-r}d_1$}{
    $x \leftarrow d(t_1^{\rm s}, s)$\;
    \ForAll{$t_3 \in T(s) \setminus \{t_1\}$ such that $(1 - 2r) t_3^{\rm d} + (1 - r) d(t_3^{\rm e}, t_1^{\rm e}) \leqq (r - 1) x + r d_1$}{
        $d_3 \leftarrow t_3^{\rm d}$\;
        $z \leftarrow d(t_3^{\rm e}, t_1^{\rm e})$\;
        \ForAll{$s' \in S$ such that $d(t_1^{\rm s}, s') + (1 - r) d(s', t_3^{\rm s}) \leqq rd_1 + (2r - 1) d_3 - (1 - r) z$}{
            $x_1 \leftarrow d(t_1^{\rm s}, s')$\;
            $x_2 \leftarrow d(s', t_3^{\rm s})$\;
            \ForAll{$t_2 \in T(s') \setminus \{t_1, t_3\}$ such that $t_2^{\rm d} \geqq \frac{1}{r}(x_1 + x_2 + z) - d_1 + \frac{1 - r}{r} d_3$}{
                $d_2 \leftarrow t_2^{\rm d}$\;
                $z_2 \leftarrow d(t_3^{\rm e}, t_2^{\rm e})$\;
                $z_1 \leftarrow d(t_2^{\rm e}, t_1^{\rm e})$\;
                \If{$\frac{x_1 + x_2 + d_3 + z_2 + z_1}{d_1 + d_2 + d_3} \leqq r$}{
                    \If{$|H| = k$}{
                        Remove the minimum element from $H$\;
                    }
                    \vspace{0mm} Add $(t_1, t_2, t_3)$ to $H$ with priority $\frac{x_1 + x_2 + d_3 + z_2 + z_1}{d_1 + d_2 + d_3} \leqq r$ \;
                    \If{$|H| = k$}{
                        $r \leftarrow$ Refer to the maximum key of $H$ without deleting it\;
                    }
                }
            }
        }
    }
}
\caption{A brute-force top-$k$ search with pruning}
\label{Alg:dynamic}
\end{algorithm2e}
%%%%%%%%%%%%%%%%%%%%%%%%%%%%%%%%%%%%%%%%%%%%%%%%%%%%%%%%%%%%%%%%%%%%%%%%%%%%%%
%
\section{Numerical Experiments}
%
%%%%%%%%%%%%%%%%%%%%%%%%%%%%%%%%%%%%%%%%%%%%%%%%%%%%%%%%%%%%%%%%%%%%%%%%%%%%%%
In the same way as computational experiments in Kira et al.~\cite{Kira2023}, we use 
approximately 17,000 real, anonymized transport lane data ($|B| = 4828, \;\, |T| = 16957$) 
across Japan, and 
we conduct experiments to enumerate efficient mixed transports. 
First we create 1000 problems (1000 matching requests) 
by randomly selecting 1000 lanes from $T$ and fixing each lane as the first lane. 
These problems are used to compare three algorithms. 
The desired reduction rate $r$ is varied from 0.35 to 0.60 in 0.05 increments. 
We implemented the three algorithms using Cython~\cite{cython, cython-guide} ($\neq$ Python) and executed them on a desktop PC
with an Intel\textregistered~Core\texttrademark~i9-9900K processor and 64GB memory installed. 
The results are presented in Tables~\ref{Tbl:result1}.

\begin{table}[h]
\centering
\caption{Computational time to process 1000 matching requests (seconds)} 
\label{Tbl:result1}
\begin{tabular}{|c|r|r|r|} \hline
\multirow{2}{*}{Scenario \textbackslash Algorithm}  & Brute-force search & pruning & $k$ Best solutions \\
& Algorithm~\ref{Alg:simple}  & Algorithm~\ref{Alg:pruning}  & Algorithm~\ref{Alg:dynamic}  \\ \hline
$r = 0.60$ & 124230.3 & 563.2 & 17.1  \\ 
$r = 0.55$ & (34h 31min) & 296.0 & 12.9  \\ 
$r = 0.50$ & $\downarrow$  \quad & 85.5 & 10.1  \\ 
$r = 0.45$ & $\downarrow$ \quad & 23.8 & 5.8  \\ 
$r = 0.40$ & $\downarrow$ \quad & 5.1 & 2.4  \\ 
$r = 0.35$ & $\downarrow$ \quad & 0.4 & 0.3  \\ \hline
\end{tabular}
\end{table}

In the case of $r = 0.60$, 
Algorithm~\ref{Alg:dynamic} is more than 7,000 times faster than simple brute force. 
In the case of $r = 0.35$, Algorithm~\ref{Alg:dynamic} is more than 400,000 times faster than simple brute force. 
%%%%%%%%%%%%%%%%%%%%%%%%%%%%%%%%%%%%%%%%%%%%%%%%%%%%%%%%%%%%%%%%%%%%%%%%%%%%%%%
%
\section{Summary}
%
%%%%%%%%%%%%%%%%%%%%%%%%%%%%%%%%%%%%%%%%%%%%%%%%%%%%%%%%%%%%%%%%%%%%%%%%%%%%%%%
In this study, we have focused on a form of joint transportation called mixed transportation 
and have proposed an algorithm for instantly enumerating set of all combinations of transport lanes with 
a reduction rate is a desired value of less. 
The results demonstrated that the proposed method is 7,000 times faster than simple brute force. 
%%%%%%%%%%%%%%%%%%%%%%%%%%%%%%%%%%%%%%%%%%%%%%%%%%%%%%%%%%%%%%%%%%%%%%%%%%%%%%%
%
\section*{Acknowledgments}
%
%%%%%%%%%%%%%%%%%%%%%%%%%%%%%%%%%%%%%%%%%%%%%%%%%%%%%%%%%%%%%%%%%%%%%%%%%%%%%%% 
%The authors would like to thank the anonymous referees for helpful comments and suggestions on this manuscript. 
This research and development was carried out as the NEDO-funded project,  titled \lq\lq 
Cross-industry joint transportation matching service to realize white logistics.'' 
In addition, Akifumi Kira was supported in part by JSPS KAKENHI Grant Numbers 17K12644, 21K11766, and 	23H01030, Japan. 
%\par
%The authors would like to thank Prof. Naoyuki Kamiyama, Prof. Katsuki Fujisawa, and Prof. Hidefumi Kawasaki of Kyushu University for their support as advisors in advancing research and development of the NEDO-funded project.
%We would like to thank Editage (www.editage.com) for English language editing.

%%%%%%%%%%%%%%%%%%%%%%%%%%%%%%  References  %%%%%%%%%%%%%%%%%%%%%%%%%%%%%%%%%%
%\bibliographystyle{jplain}
%\bibliography{bib}

\begin{thebibliography}{99}
%------------------------------------------------------
\bibitem{cython}
Behnel,~S., Bradshaw,~R., Citro,~C., Dalcin,~L., Seljebotn,~D.S., \& Smith,~K. (2010). 
Cython: The best of both worlds. 
{\it Computing in Science \& Engineering}, {\bf 13}(2), 31--39.

\bibitem{policy-a}
Government of Japan (2017). %(Ministry of Land, Infrastructure, Transport and Tourism) website. 
Comprehensive Physical Distribution Policy (FY2017--FY2020). (in Japanese), Accessed 2025 Mar 31. \\
\url{https://www.mlit.go.jp/seisakutokatsu/freight/seisakutokatsu_freight_tk1_000128.html}

\bibitem{policy-b}
Government of Japan (2021). %(Ministry of Land, Infrastructure, Transport and Tourism) website. 
Comprehensive Physical Distribution Policy (FY2021--FY2025). (in Japanese), Accessed 2025 Mar 31. \\
\url{https://www.mlit.go.jp/seisakutokatsu/freight/butsuryu03100.html}

\bibitem{JPR2021}
Japan Pallet Rental Corporation and Gunma University (2021). 
Rapid listing of combinations of transport routes with high cooperative effects --- 
Development of joint transport matching technology, 
Joint Press Release (October 21, 2021), Accessed 2025 Mar 31. \\
\url{https://www.gunma-u.ac.jp/wp-content/uploads/2021/10/Release20211021_EN.pdf}

\bibitem{Kira2021}
Kira, A., Terajima, N., \& Watanabe, Y. (2021).
Transport combination enumeration program, transport combination enumeration method and transport combination enumeration system. 
Japanese Patent No. 7373169 (October 25, 2023), 
Japanese Patent Application No. 2021-171440 (October 20, 2021).
 
\bibitem{Kira2023}
Kira, A., Terajima, N., Watanabe, Y., Yamamoto, H. (2023).
Shipper collaboration matching: fast enumeration of triangular transports with high cooperation effects. 
{\it arXiv}:2303.18222, 16 pages.

\bibitem{cython-guide}
Smith,~K.W. (2015). 
{\it Cython -- A guide for python programmers}.  
O'Reilly Media, Inc.

\bibitem{heap}
Wiliams,~J.W.J. (1964)  
Algorithm 232: Heapsort. 
{Communications of the ACM}
{\bf 7}(6),
347--348.
\end{thebibliography}
%------------------------------------------------------ 

%%%%%%%%%%%%%%%%%%%%%%%%%%%%%%%%%%%%%%%%%%%%%%%%%%%%%%%%%%%%%%%%%%%%%%%%%%%%%%

%%%%%%%%%%%%%%%%%%%%%%%%%%%%%%%%%%%%%%%%%%%%%%%%%%%%%%%%%%%%%%%%%%%%%%%%%%%%%%%
\end{document}